\documentclass[12pt]{article}

\usepackage{hyperref}

\usepackage{amsfonts, amsmath, amssymb, amsthm}
\usepackage{a4}

\usepackage{upgreek}

\DeclareMathOperator{\const}{const}
\DeclareMathOperator{\minor}{minor}
\DeclareMathOperator{\Pfaffian}{Pfaffian}
\DeclareMathOperator{\linearspan}{span}

\newtheorem{theorem}{Theorem}
\newtheorem*{conjecture}{Conjecture}
\newtheorem{lemma}{Lemma}

\theoremstyle{definition}
\newtheorem{dfn}{Definition}
\newtheorem{mpt}{Assumption}

\theoremstyle{remark}
\newtheorem{remark}{Remark}
\newtheorem{example}{Example}

\author{I.G. Korepanov}
\title{Free fermions on a piecewise linear four-manifold. II:~Pachner moves}
\date{May 2016 -- March 2017}

\begin{document}

\sloppy

\maketitle

\begin{abstract}
This is the second in a series of papers where we construct an invariant of a four-dimen\-sional piecewise linear manifold~$M$ with a given middle cohomology class $h\in H^2(M,\mathbb C)$. This invariant is the square root of the torsion of unusual chain complex introduced in Part~I of our work, multiplied by a correcting factor. Here we find this factor by studying the behavior of our construction under all four-dimen\-sional Pachner moves, and show that it can be represented in a multiplicative form: a product of same-type multipliers over all 2-faces, multiplied by a product of same-type multipliers over all pentachora.
\end{abstract}

\section{Introduction}\label{s:iintro}

This is the continuation of paper~\cite{part-I}; we also call this latter `Part~I', while the present paper `Part~II'. We refer to formulas, definitions, etc.\ from Part~I in the following format: formula~(I.5) means formula~5 in Part~I, etc.

\medskip

A standard way to define a piecewise linear (PL) manifold is via its \emph{triangulation}. In this paper, we will be dealing with \emph{four-dimen\-sion\-al} manifolds, and a triangulation of such manifold means that it is represented as a union of 4-simplices, also called \emph{pentachora}, glued together in a proper way that can be described purely combinatorially. An \emph{invariant} of such manifold is a quantity that may need an actual triangulation for its calculation, but must \emph{not} depend on this triangulation.

A theorem of Pachner~\cite{Pachner} states that a triangulation of a PL manifold can be transformed into any other triangulation using a finite sequence of \emph{Pachner moves}; monograph~\cite{Lickorish} can be recommended as a pedagogical introduction to this subject. And, as indicated in~\cite[Section~1]{Lickorish}, in order to construct invariants of PL manifolds, it makes sense to construct algebraic relations corresponding to Pachner moves, also called their \emph{algebraic realizations}.

It turns out that very interesting mathematical structures appear if we begin constructing a realization of four-dimen\-sion\-al Pachner moves by ascribing a \emph{Grassmann--Gaussian weight} to each pentachoron. Here `Gaussian' means that this weight is proportional to the exponential of a quadratic form, and `Grassmann' means that this form depends on \emph{anticommuting} Grassmann variables. Each Grassmann variable is supposed to live on a 3-face (tetrahedron) of a pentachoron, and gluing two pentachora along a 3-face corresponds to \emph{Berezin integration} w.r.t.\ the corresponding variable. A large family of such realizations for Pachner move 3--3 was discovered in paper~\cite{KS2}, and then a full parameterization for (a Zariski open set of) such relations was found in~\cite{full-nonlinear}. A beautiful fact is that this \emph{nonlinear} parameterization goes naturally in terms of a \emph{2-cocycle} given on both initial and final configurations of the Pachner move; we call these respective configurations (clusters of pentachora) the \emph{left-} and \emph{right-hand side} (l.h.s.\ and r.h.s.) of the move.

There was, however, one unsettled problem with the realization of move 3--3 in~\cite{full-nonlinear}: not all involved quantities were provided with their explicit expressions in terms of the 2-cocycle. In particular, Theorem~9 in~\cite{full-nonlinear} was just an existence theorem for the proportionality coefficient between the Berezin integrals representing the l.h.s.\ and r.h.s.\ of the move, while this coefficient is crucial for constructing an invariant for a whole `big' manifold. Also, it remained to find realizations for the rest of Pachner moves, namely 2--4 and 1--5.

One possible solution to the problem with the coefficient was proposed in~\cite{gce}, in a rather complicated way combining computational commutative algebra with a guess-and-try method, and leaving the feeling that the algebra behind it deserves more investigation.

It turns out that a more transparent way to solving the mentioned problem with the coefficient appears if we consider this problem \emph{together} with constructing formulas for moves 2--4 and 1--5. This is what we are doing in the present paper: we provide all necessary formulas for coefficients, and in a multiplicative form suitable for `globalizing', that is, transition to a formula for the whole manifold. To be exact, we construct an invariant of a pair $(M,h)$, where $M$ is a four-dimen\-sional piecewise linear manifold, and $h\in H^2(M,\mathbb C)$ is a given middle cohomology class.

One algebraic problem still remains unsolved though, namely, a formal proof of what we have to call `Conjecture', see Subsection~\ref{ss:ER}, and what is actually a firmly established mathematical fact. That is simply a formula involving ten indeterminates (over field~$\mathbb C$), whose both sides are composed using the four arithmetic operations and also square root signs and parentheses. The experience gained in the previous work~\cite{rels,KS2} led the author to the idea that such formula exists---and indeed, the reader can check it on a computer by substituting any random numerical values for the indeterminates\footnote{For instance, using the program code available from the author and written for \emph{Maxima} computer algebra system.}. Many such checks have been already carried out, but the available computer capabilities are not enough to check our formula symbolically\footnote{Currently, efforts to solve this problem are being made mainly in two following directions: write a \emph{specialized} software that would be able to handle more effectively our specific expressions, and/or---of course!---discover their new properties that would enable us to find a `conceptual' proof. Also, it can be proved that if such a formula has been checked (using exact arithmetic) for a huge enough set of tuples of arguments, than it is right for \emph{all} arguments.}.

Below,
\begin{itemize}\itemsep 0pt
 \item in Section~\ref{s:P}, we recall the Pachner moves in four dimensions, and also express the moves 2--4 and 1--5 in terms of 3--3 and two kinds of auxiliary moves `0--2'. Introducing these auxiliary moves makes our reasonings and formulas simple and transparent,
 \item in Section~\ref{s:u}, we present the general structure of our manifold invariant, and explain how it is expressed in terms of Grassmann--Berezin calculus of anticommuting variables,
 \item in Section~\ref{s:33}, we show how it follows from the `local' Grassmann-algebraic relation 3--3 that our proposed `global' invariant is indeed invariant under moves 3--3. This is a general theorem, where some important quantities, called~$\eta_u$, remain, at the moment, unspecified,
 \item in Section~\ref{s:oF}, we provide some formulas needed for proofs of two theorems in the next Section~\ref{s:02}. These formulas have also an algebraic beauty of their own,
 \item in Section~\ref{s:02}, we consider Grassmann-algebraic realizations of the mentioned moves 0--2. These turn out to have a simple and elegant form, involving Grassmann delta functions. Moreover, it turns out that they produce expressions for the mentioned quantities~$\eta_u$,
 \item and in Section~\ref{s:f}, we prove that our invariant---built initially using a 2-cocycle~$\omega$ on a PL manifold~$M$---depends actually only on the cohomology class $h\ni \omega$. Also, we explain its independence from such things as the signs of square roots~$K_t$ (see~(I.36)) that appeared in our calculations. We thus have (assuming our Conjecture) indeed an invariant of the pair $(M,h)$, as was promised in the Introduction to Part~I~\cite{part-I} of this work.
\end{itemize}

\section{Pachner moves and some useful decompositions of them}\label{s:P}

\subsection{Pachner moves}\label{ss:P}

A four-dimensional Pachner move replaces a cluster of pentachora in a manifold triangulation by another cluster occupying the same place in the triangulation. As already said in the Introduction, we call the initial and final clusters of a move its left- and right-hand side, respectively. There are five kinds of Pachner moves in four dimensions.

In this Subsection, pentachora and other simplices are determined by their vertices. We will describe Pachner moves using a fixed numeration of vertices, and we will be using this numeration throughout this paper. All pentachora must be \emph{oriented} consistently; by default, the orientation of a pentachoron corresponds to the order of its vertices, or, if it has the opposite orientation, this is marked by a wide tilde above, as in `pentachoron~$\widetilde{12346}$'.

\begin{remark}
Starting from Subsection~\ref{ss:24}, it will be convenient for us to consider triangulations in a broader sense, where there may be more than one simplex with the same vertices. We will be using tildes or primes to distinguish between such simplices.
\end{remark}

\paragraph{Move 3--3} transforms a cluster of three pentachora into a different cluster, also of three pentachora, as follows:
\begin{equation}\label{P33}
12345, \widetilde{12346}, 12356 \rightarrow 12456, \widetilde{13456}, 23456.
\end{equation}
Also, the l.h.s.\ of this move has 2-face~$123$ not present in the r.h.s., while the r.h.s.\ has 2-face~$456$ not present in the l.h.s. The two sides of this move differ also in their \emph{inner} tetrahedra: these are $1234$, $1235$ and~$1236$ in the l.h.s., and $1456$, $2456$ and~$3456$ in the r.h.s. All this information will be included in our algebraic realization~\eqref{33} of this move.

\paragraph{Move 2--4} transforms a cluster of two pentachora into a cluster of four pentachora, as follows:
\begin{equation}\label{P24}
\widetilde{13456}, 23456 \rightarrow 12345, \widetilde{12346}, 12356, \widetilde{12456}.
\end{equation}
One more (kind of) Pachner move is the inverse to~\eqref{P24}.

\paragraph{Move 1--5} transforms just one pentachoron into a cluster of five pentachora, as follows:
\begin{equation}\label{P15}
23456 \rightarrow 12345, \widetilde{12346}, 12356, \widetilde{12456}, 13456.
\end{equation}
One more (kind of) Pachner move is the inverse to~\eqref{P15}.

\bigskip

There are some decompositions of moves 2--4 and 1--5 whose importance will be seen in Section~\ref{s:02}. Namely, we will represent move 2--4 as a composition of what we call `first move 0--2' and Pachner move 3--3. This is explained in Subsection~\ref{ss:24}. Similarly, Pachner move 1--5 will be expressed, in Subsection~\ref{ss:15}, as a composition of `second move 0--2' and Pachner move 2--4.

Hence, in order to prove that a quantity is a PL manifold invariant, it is enough to prove its invariance under moves 3--3 and the mentioned two kinds of moves 0--2.

\subsection{Inflating two adjacent tetrahedra into a four-dimensional pillow, and Pachner move 2--4}\label{ss:24}

\begin{dfn}\label{dfn:1}
The \emph{first move 0--2} is defined as follows. Consider two tetrahedra $1456$ and~$2456$ having the common 2-face~$456$. We are going to glue to them two pentachora with opposite orientations, that is, $\widetilde{12456}$ and $12456$. After gluing the first of these, the ``free part'' of its boundary consists of three tetrahedra $1245$, $1246$ and~$1256$, and then we glue the second pentachoron to these three tetrahedra, thus obtaining a ``pillow'' whose boundary consists of two copies of tetrahedron~$1456$ and two copies of~$2456$, and having the inner edge~$12$.
\end{dfn}

Consider now the left-hand side of Pachner move 2--4~\eqref{P24} consisting of pentachora $\widetilde{13456}$ and~$23456$. We inflate the part $1456 \cup 2456$ of its boundary into the pillow described above, and get the pentachora $\widetilde{12456}$, $12456$, $\widetilde{13456}$ and~$23456$. Then we do the \emph{inverse} to move~\eqref{P33} on the last three of them, and obtain the r.h.s.\ of~\eqref{P24}.

It will also be important for us what happens with \emph{2-faces} when the pillow from Definition~\ref{dfn:1} is inserted into a triangulation instead of just two tetrahedra $1456$ and $2456$. It can be checked that there appear three new 2-faces, namely $123$, $124$ and~$125$, while the face~$456$ is \emph{doubled}. This information (about 2- as well as 3-faces) will be naturally included into our algebraic realization \eqref{P}, \eqref{pw24} of the first move 0--2.

\subsection{Inflating a single tetrahedron into a four-dimensional pillow, and Pachner move 1--5}\label{ss:15}

\begin{dfn}\label{dfn:2}
The \emph{second move 0--2} is defined as the inflation of a single tetrahedron~$3456$ into the pillow made of pentachora $13456$ and $\widetilde{13456}$ glued along four pairs of their same-name 3-faces containing new vertex~$1$.
\end{dfn}

Consider now the l.h.s.\ of~\eqref{P15}, i.e., the pentachoron~$23456$, and inflate its 3-face~$3456$ this way. We get the pentachora $13456$, $\widetilde{13456}$ and~$23456$. Then we apply the move 2--4~\eqref{P24} to the last two of them, and get the r.h.s.\ of~\eqref{P15}.

\section{Structure of the invariant}\label{s:u}

Here, in Subsection~\ref{ss:tor}, we describe our invariant as the \emph{square root of exotic torsion with a correcting multiplier}. In fact, we give it almost full definition below in formula~\eqref{inv}, where we leave undefined (until Subsection~\ref{ss:02_24}) only factors~$\eta_u$ attached to all pentachora~$u$. Our tool for studying the behavior of our invariant under Pachner moves will be Grassmann--Berezin calculus of anticommuting variables, so then, in Subsection~\ref{ss:igb}, we rewrite our invariant in terms of this calculus.

\subsection{Invariant: square root of exotic torsion with a correcting multiplier}\label{ss:tor}

We want to use the \emph{torsion} of chain complex~(I.11) in the construction of our invariant. So, we adopt the following assumption.

\begin{mpt}\label{mpt:e}
The complex~(I.11) is \emph{acyclic}.
\end{mpt}

\begin{remark}
Assumption~\ref{mpt:e} does hold for some manifolds, namely, for instance, for sphere~$S^4$. Note, however, that a torsion may be constructed even if Assumption~\ref{mpt:e} fails, see~\cite[Subsection~3.1]{Turaev}.
\end{remark}

Recall that all linear spaces in complex~(I.11) are equipped with distinguished bases, hence all linear mappings are identified with their matrices. According to the general theory~\cite[Subsection~2.1]{Turaev}, the torsion of complex~(I.11) is
\begin{equation}\label{tauu}
\tau=\frac{\minor f_1\cdot \minor f_3\cdot \minor f_5}{\minor f_2\cdot \minor f_4},
\end{equation}
where the minors are chosen as follows. For each of the six nonzero linear spaces in~(I.11), a subset~$\mathfrak b_i$ is taken of its basis, $i=1,\ldots,6$, and these subsets must satisfy the following conditions:
\begin{itemize}\itemsep 0pt
 \item $\mathfrak b_1$ is the empty set,
 \item $\mathfrak b_{i+1}$ contains the same number of basis vectors as $\overline{\mathfrak b_i}$ (the complement to the $i$-th subset),
 \item submatrices of $f_i$, \ $i=1,\ldots,5$, whose rows correspond to $\mathfrak b_{i+1}$ and columns to $\overline{\mathfrak b_i}$, are nondegenerate.
\end{itemize}
Such subsets~$\mathfrak b_i$ always exist for an acyclic complex, and minors in~\eqref{tauu} are by definition the determinants of the mentioned submatrices of $f_i$.

\begin{remark}
$\minor f_1$ is of course just one matrix element of one-column matrix~$f_1$; similar statement applies to $\minor f_5$ as well.
\end{remark}

Using the symmetry of complex~(I.11) (see the text right after formula~(I.11)), we can choose all the minors in~\eqref{tauu} in a symmetric way, namely, so that the rows used in $\minor f_1$ have the same numbers as the columns in $\minor f_5$; similarly for $f_2$ and~$f_4$; and the submatrix of~$f_3$ used for the minor will involve rows and columns with the same (subset of) numbers and will thus be skew-symmetric, like~$f_3$ itself. Using also the notion of \emph{Pfaffian}, we can write:
\begin{equation}\label{rt}
\sqrt{\tau} = \frac{\minor f_1\cdot \Pfaffian(\mathrm{submatrix\;of\;} f_3)}{\minor f_2}.
\end{equation}

We will show that a PL manifold invariant can be obtained from our quantity~\eqref{rt} if it is multiplied by a correcting factor, introduced to ensure the invariance under Pachner moves. Namely, our invariant will be
\begin{equation}\label{inv}
I = \prod_{\substack{\text{all}\\ \!\!\!\!\!\text{2-faces }s\!\!\!\!\!}} q_s^{-1} \prod_{\substack{\text{all}\\ \!\!\!\!\!\text{pentachora }u\!\!\!\!\!}} \eta_u \cdot \sqrt{\tau}.
\end{equation}
Recall~(I.34) that $q_s=\sqrt{\omega_s}$ are square roots of 2-cocycle~$\omega$'s values. Likewise, each quantity~$\eta_u$, attached to pentachoron~$u$, will be expressed algebraically in terms of values~$\omega_s$ for $s\subset u$, but the exact definition of~$\eta_u$ needs some preparational work; it will appear in Subsection~\ref{ss:02_24} as formula~\eqref{eta}.

Formula~\eqref{inv} implies that the triangulation vertices have been ordered, see Convention~I.1 (and for independence of~$I$ of this ordering, see Theorem~\ref{th:oe} below). It is also implied that the values of~$q_s$ and~$K_t$~(I.36) must agree in the sense of formula~(I.40); it makes sense to reproduce it here:
\begin{equation}\label{I.40}
K_{ijkl} = q_{ijk}q_{ijl}q_{ikl}q_{jkl},
\end{equation}
for any tetrahedron $t=ijkl$ with $i<j<k<l$.

\begin{remark}\label{r:+-}
As the reader can see, formula~\eqref{inv} contains square roots and, besides, it looks not easy to specify the order of rows and/or columns in the minors in formulas like~\eqref{rt}. Below, in terms of Grassmann--Berezin calculus, we will get the same problem disguised as the choice of integration order for multiple Berezin integrals. This leads to the agreement that we consider~$I$ as determined \emph{up to a sign}\footnote{Which is not very surprising for a theory dealing with something like Reidemeister torsion.}, so we will not pay much attention to the signs in our formulas, as long as these signs can affect only the sign of~$I$. A subtler situation---namely with \emph{fourth} roots---is considered in Theorem~\ref{th:oe}. 
\end{remark}

\subsection{Invariant in terms of Grassmann--Berezin calculus}\label{ss:igb}

In order to study the behavior of quantity~\eqref{rt} under Pachner moves, it makes sense to express it in terms of Grassmann--Berezin calculus. Recall (Definition~I.1) that our Grassmann algebras are over field~$\mathbb C$, hence `linear' means `$\mathbb C$-linear', etc. As we will see in formulas \eqref{33}, \eqref{P} and~\eqref{pw15}, the algebraic operation corresponding to gluing pentachora together will be, in our construction, the \emph{Berezin integral}, so we recall now its definition.

\begin{dfn}\label{dfn:Bi}
\emph{Berezin integral}~\cite{B1,B2} with respect to a generator~$\vartheta$ of a Grassmann algebra~$\mathcal A$ is a linear functional
\[
f\mapsto \int f\, \mathrm d\vartheta
\]
on~$\mathcal A$ defined as follows. As $\vartheta^2=0$, any algebra element~$f$ can be represented as $f=f_0+f_1\vartheta$, where none of $f_0$ and~$f_1$ contain~$\vartheta$. By definition,
\[
\int f\, \mathrm d\vartheta = f_1.
\]
\emph{Multiple integral} is defined as the iterated one:
\[
\iint f\, \mathrm d\vartheta_1\, \mathrm d\vartheta_2 = \int \left(\int f\, \mathrm d\vartheta_1\right) \mathrm d\vartheta_2.
\]
\end{dfn}

We will need only a few simple properties of Berezin integral; we recall them when necessary. Right now, we mention the following two properties:
\begin{itemize} \itemsep 0pt
 \item Berezin integral looks very much like the (left) \emph{derivative} (see Definition~I.3). In fact, there is an operator in Grassmann algebras called \emph{right derivative} that coincides exactly with Berezin integral. Also, some authors simply define the Berezin integral to be the same as the \emph{left} derivative; we will, however, stick to the traditional Definition~\ref{dfn:Bi}. Anyhow, the results of integration and differentiation of either \emph{even} or \emph{odd} (all monomials have even or odd degrees, respectively) Grassmann algebra element can differ at most by a sign---and we will make use of this fact when proving Lemma~\ref{l:GB} below,
 \item integral of a Grassmann--Gaussian exponential is
\begin{equation}\label{pfi}
\int \exp(\uptheta^{\mathrm T}A\,\uptheta) \prod_{i=1}^n \mathrm d\vartheta_i = (-2)^{n/2} \Pfaffian A,
\end{equation}
where $A$ is a skew-symmetric matrix (with entries in~$\mathbb C$), and $\uptheta$ is a column of Grassmann generators: $\uptheta=\begin{pmatrix}\vartheta_1 & \dots & \vartheta_n\end{pmatrix}^{\mathrm T}$.
\end{itemize}

Recall now that in Subsection~I.4.2 we have put a Grassmann variable (generator)~$\vartheta_t$ in correspondence to each tetrahedron~$t$ in the triangulation. Below ``Grassmann algebra'' means, by default, the algebra generated by all these~$\vartheta_t$. Moreover, we have introduced, in Subsection~I.6.1, the column~$\Uptheta=\begin{pmatrix} \vartheta_1 & \ldots & \vartheta_{N_3} \end{pmatrix}^{\mathrm T}$  made of all these~$\vartheta_t$.

We now also introduce the \emph{row} consisting of all (left) \emph{differentiations} with respect to variables~$\vartheta_t$:
\begin{equation}\label{dt}
\mathbf D=\begin{pmatrix} \partial_1 & \ldots & \partial_{N_3} \end{pmatrix}, \quad \text{where \ } \partial_t=\frac{\partial}{\partial \vartheta_t}.
\end{equation}
Consider the product $\mathbf D f_2$, where elements of both row~$\mathbf D$ and matrix~$f_2$ are understood as $\mathbb C$-linear operators acting in our Grassmann algebra, that is, differentiations and multiplications by constants, respectively. This way, $\mathbf D f_2$ makes a row of differential operators. Recall now that columns of~$f_2$---and thus elements of row~$\mathbf D f_2$---correspond to basis elements in $Z^2(M,\mathbb C)$---the second space in complex~(I.11), and these include elements corresponding to edges in a set~$\mathsf B$ whose complement~$\overline{\mathsf B}$ makes a maximal tree in the 1-skeleton of our triangulation, see the paragraph right between Remarks I.3 and~I.4. We call the element in~$\mathbf D f_2$, corresponding this way to an edge~$b\in \mathsf B$, \emph{global edge operator} corresponding to~$b$, and denote it~$D_b$. Clearly, $D_b$ does not depend on a specific~$\mathsf B$, as long as $\mathsf B\ni b$.

\begin{remark}
Row~$\mathbf D f_2$ includes also elements corresponding to a pullback of some chosen basis in $H^2(M,\mathbb C)$, see the paragraph right before Remark~I.2.
\end{remark}

\begin{remark}
In the proof of Theorem~I.7, we denoted~$D_c$ the differential operator corresponding to a \emph{2-cochain}~$c$. In the notations of that proof, our operators~$D_b$ correspond to $c=\delta b$ and should be called~$D_{\delta b}$. Hopefully, our less pedantical notations will bring no confusion.
\end{remark}

Denote now the submatrix of matrix~$f_2$ consisting exactly of its columns used in $\minor f_2$ in~\eqref{rt} as~$\tilde f_2$ (these are of course all columns except one whose number is the same as the number of row used in $\minor f_1$), and consider the product $\mathbf D \tilde f_2$, consisting of all operators in~$\mathbf D f_2$ except one. We denote by $\boldsymbol{\partial}$ the \emph{product of all operators in~$\mathbf D \tilde f_2$}.

\begin{lemma}\label{l:GB}
In the notations of formula~\eqref{rt},
\begin{equation}\label{23GB}
\frac{\Pfaffian(\mathrm{submatrix\;of\;} f_3)}{\minor f_2} = 2^{-m_3/2} \idotsint \boldsymbol{\partial}^{-1} 1 \cdot \exp(\Uptheta^{\mathrm T} f_3 \,\Uptheta) \prod_{\substack{\mathrm{all}\\ \!\!\!\!\mathrm{tetrahedra\; }t\!\!\!\!}} \mathrm d\vartheta_t,
\end{equation}
where $\boldsymbol{\partial}^{-1} 1$ is \emph{any} such Grassmann algebra element~$w$ that $\boldsymbol{\partial}w = 1$, and $m_3$ is the size of the $\mathrm{submatrix\;of\;} f_3)$.
\end{lemma}

\begin{proof}
It is not hard to see that $\boldsymbol{\partial}^{-1} 1$ can be chosen as one monomial
\[
\boldsymbol{\partial}^{-1} 1 = \frac{\prod_{t\in\mathfrak b_3} \vartheta_t}{\minor f_2},
\]
where we use the notations introduced after formula~\eqref{tauu}; in other words, the product in the numerator goes over the tetrahedra corresponding to those rows of~$f_2$ that enter in~$\minor f_2$. With this choice, a small exercise in Grassmann--Berezin calculus using formula~\eqref{pfi} shows that \eqref{23GB} holds indeed.

What remains is to show that the r.h.s.\ of~\eqref{23GB} does \emph{not} depend on a choice of $\boldsymbol{\partial}^{-1} 1$. In other words, we must show that if $w_0$ is such that $\boldsymbol{\partial} w_0 =0$, then
\begin{equation}\label{w0}
\idotsint w_0 \exp(\Uptheta^{\mathrm T} f_3 \,\Uptheta) \prod_{\substack{\mathrm{all}\\ \!\!\!\!\mathrm{tetrahedra\; }t\!\!\!\!}} \mathrm d\vartheta_t = 0.
\end{equation}
First, we note that we defined $\boldsymbol{\partial}$ as a product of differential operators---linear combinations of differentiations w.r.t.\ Grassmann generators~$\vartheta_t$---each of which annihilates $\exp(\Uptheta^{\mathrm T} f_3 \,\Uptheta)$, as a formula proved within the proof of Theorem~I.7 tells us; as it happens to be unnumbered, it makes sense to repeat it here:
\[
D_c \exp(\Uptheta^{\mathrm T} f_3 \,\Uptheta) = D_c \prod_u \mathcal W_u = 0.
\]
Using Leibniz rule~(I.3) repeatedly, we deduce (from the fact that the operators annihilate the exponential) that
\begin{equation}\label{d0}
\boldsymbol{\partial} \bigl(w_0 \exp(\Uptheta^{\mathrm T} f_3 \,\Uptheta)\bigr) = 0.
\end{equation}

Finally, the multiple integral in~\eqref{w0} is equivalent to ($\pm$) the product of anticommuting differentiations~$\partial_t$ for all tetrahedra~$t$. Following our Subsection~I.2.2, we denote $U^*=\linearspan_{\mathbb C}\{\partial_t\}$ the linear space generated by these  differentiations. On the other hand, $\boldsymbol{\partial}$ is the product of some linearly independent linear combinations~$D_c$ of~$\partial_t$. Enlarging the set of these~$D_c$ to a basis in~$U^*$ by adding the necessary number of new operators, and denoting the product of these latter as~$\boldsymbol D$, we see that the integration in~\eqref{w0} is equivalent to $\const\cdot \boldsymbol{\partial D}$. Hence, \eqref{w0} immediately follows from~\eqref{d0}.
\end{proof}

\section{Grassmann-algebraic relation 3--3: how it works within a triangulation}\label{s:33}

In this Section we, while still not specifying what exactly the quantities~$\eta_u$ are in formula~\eqref{inv}, prove the following theorem showing that the `local' relation~\eqref{33} implies that the `global' quantity is indeed invariant \emph{under moves 3--3}. Relation~\eqref{33} corresponds to the move 3--3 as it was described in Subsection~\ref{ss:P}.

\begin{remark}
The explicit form of~$\eta_u$ satisfying~\eqref{33} will appear very naturally in Section~\ref{s:02} while examining the moves 0--2 introduced in Subsections \ref{ss:24} and~\ref{ss:15}. The proof of~\eqref{33}, with these~$\eta_u$, turns out to be a specific algebraic problem, as we already said in the Introduction; see Conjecture in Subsection~\ref{ss:ER}.
\end{remark}

\begin{theorem}\label{th:33}
Let pentachora $12345$, $\widetilde{12346}$ and~$12356$ be contained in the triangulation of manifold~$M$, and consider move 3--3~\eqref{P33} done on them. Suppose the following Grassmann-algebraic realization of this move holds:
\begin{multline}\label{33}
\frac{\eta_{12345}\, \eta_{12346}\, \eta_{12356}}{q_{123}} \iiint \mathcal W_{12345} \widetilde{\mathcal W}_{12346} \mathcal W_{12356} \,\mathrm d\vartheta_{1234} \,\mathrm d\vartheta_{1235} \,\mathrm d\vartheta_{1236} \\
= \frac{\eta_{12456}\, \eta_{13456}\, \eta_{23456}}{q_{456}} \iiint \mathcal W_{12456} \widetilde{\mathcal W}_{13456} \mathcal W_{23456} \,\mathrm d\vartheta_{1456} \,\mathrm d\vartheta_{2456} \,\mathrm d\vartheta_{3456} \,,
\end{multline}
where we write $\widetilde{\mathcal W}_{12346}$ instead of $\mathcal W_{\widetilde{12346}}$, etc.

Then, quantities~I~\eqref{inv} calculated for the initial and final triangulations are the same.
\end{theorem}

\begin{proof}
Due to Lemma~\ref{l:GB}, it is enough to show that the multiple integral in the right-hand side of~\eqref{23GB} for the initial or final triangulation can be obtained by making first triple integration in the respective side of~\eqref{33}, and then further integrating in the rest of variables~$\vartheta_t$. Consider the initial triangulation for definiteness.

Consider differentiations $\partial_{1234}$, $\partial_{1235}$ and~$\partial_{1236}$; their product is, up to a possible sign, the same operation as the triple integration in the left-hand side of~\eqref{33}; the sign may appear because our differentiations are, by default, \emph{left}, while integral corresponds to \emph{right} differentiations. Let
\[
\mathbf W=\prod_u \mathcal W_u=\exp(\Uptheta^{\mathrm T} f_3 \,\Uptheta)
\]
be the product of pentachoron weights~(I.16) over all triangulation. Then $\partial_{1234} \partial_{1235} \partial_{1236} \mathbf W \ne 0$.

On the other hand, $\boldsymbol{\partial}$ is the product of differentiations in~$\mathbf D \tilde f_2$, each of which annihilates~$\mathbf W$---see the proof of Theorem~I.7. It follows then that only zero is contained in the intersection of the linear space spanned by differentiations in~$\mathbf D \tilde f_2$ with the linear space spanned by $\partial_{1234}$, $\partial_{1235}$ and~$\partial_{1236}$. It follows further that $\boldsymbol{\partial}^{-1}1$ can be chosen as a single Grassmann monomial containing none of $\vartheta_{1234}$, $\vartheta_{1235}$ and~$\vartheta_{1236}$.

Indeed, make the following matrix~$C$: its columns correspond to all tetrahedra~$t$ in the triangulation, while the $i$-th row consists of the coefficients of~$\partial_t$ in the $i$-th operator in~$\mathbf D \tilde f_2$. Let also tetrahedra $1234$, $1235$ and~$1236$ correspond to the three \emph{last} columns. Reduce then matrix~$C$ to the \emph{echelon} form; the leading element in any row cannot then belong to the three last columns---and the desired monomial can be obtained as the product of variables~$\vartheta_t$ corresponding to these exactly leading elements, with some coefficient.

Now we see that $\boldsymbol{\partial}^{-1}1$ can be chosen not to contain the integration variables in~\eqref{33}, and pentachoron weights~$\mathcal W_u$ for $u$ not present in~\eqref{33} do not contain them either. This means that all these factors can be taken out from under the integration in~\eqref{33}, and this latter can be performed first, as we wanted to show.
\end{proof}

\section[Edge operators and cocycle~$\omega$ from matrix~$F$]{Edge operators and cocycle~$\boldsymbol{\omega}$ from matrix~$\boldsymbol{F}$}\label{s:oF}

In this Section, we write out some formulas needed for proofs of Theorems \ref{th:eta} and~\ref{th:e2} below. Apart from that need, these formulas have their own intriguing algebraic beauty.

We work here within one fixed pentachoron~$u$. Let $b$ and~$t$ denote an edge and a tetrahedron such that $b\subset t\subset u$. Recall that the (local) edge operator has the following form\footnote{\eqref{I.14c} reproduces our formula~(I.14), except for one misprint in the latter: `$\partial_{bt}$' in~(I.14) must be read as~`$\partial_t$'.}:
\begin{equation}\label{I.14c}
d_b = d_b^{(u)} = \sum_{\substack{t\subset u\\ t\supset b}} (\beta_{bt}\partial_t + \gamma_{bt}\vartheta_t),
\end{equation}
and the coefficients in it can be calculated, up to a scalar factor---we call it~$h_b$---directly from the matrix~$F$ corresponding to~$u$. We have already done this calculation in~\cite[formula~(20)]{full-nonlinear}; here we write the result in a slightly different form, namely:
\begin{equation}\label{hb}
\beta_{bt}=h_b \tilde{\beta}_{bt},\qquad \gamma_{bt}=h_b \tilde{\gamma}_{bt},
\end{equation}
where
\begin{align}
\tilde{\beta}_{bt} &= F_{ik}F_{jl}-F_{il}F_{jk}, \label{bF} \\
\tilde{\gamma}_{bt} &= F_{ik}F_{jm}F_{lm}-F_{im}F_{jk}F_{lm}-F_{il}F_{jm}F_{km}+F_{im}F_{jl}F_{km}, \label{gF}
\end{align}
and the notations are as follows: $i,\ldots,m$ are vertices determined by the requirement that they must give our edge, tetrahedron and pentachoron, \emph{together with their respective orientations}, according to
\[
b=ij,\quad t=ijkl,\quad u=ijklm,
\]
(orientations correspond to the order of vertices; orientation of~$t$ is induced from~$u$); and $F_{ik}$ is a shorthand for the matrix element between the tetrahedra \emph{not} containing vertices $i$ and~$j$, respectively, that is,
\[
F_{ik} \stackrel{\mathrm{def}}{=} F_{jklm,ijlm}.
\]

Coefficients~$h_b$ for the ten edges $b\subset u$ must be chosen in such way as to ensure the 1-cycle relations~(I.21) for edge operators. It can be checked that such~$h_b$ are given by the following remarkable formula:
\begin{equation}\label{pf}
h_{ij} \stackrel{\mathrm{def}}{=} p\,(\tilde{\beta}_{ki,t}\,\tilde{\gamma}_{kj,t}+\tilde{\beta}_{kj,t}\,\tilde{\gamma}_{ki,t}),
\end{equation}
where the r.h.s.\ actually does \emph{not} depend on the tetrahedron $t\subset u$, and $p$ is an arbitrary scalar factor.

Even more remarkably, the values~$\omega_s$ of cocycle~$\omega$, determined by homogeneous linear equations~(I.22) (and the unnumbered formula right below~(I.22)), turn out to be proportional to the following triple products over the edges of tetrahedron~$s=ijk$:
\begin{equation}\label{hhh}
\omega_{ijk} \propto h_{ij}^{-1}h_{ik}^{-1}h_{jk}^{-1}.
\end{equation}
The proportionality coefficient in~\eqref{hhh} is then determined by the normalization Convention~I.5 for edge operators; this leads to the formula
\begin{equation}\label{ph}
\omega_{ijk} = \frac{p}{h_{ij}h_{ik}h_{jk}}.
\end{equation}

\begin{remark}
We have already met triple products of the kind~\eqref{hhh} in our elliptic parameterization~\cite[(50)]{full-nonlinear}; the corresponding matrix elements of~$F$ had the elegant form~\cite[(51)]{full-nonlinear}, or could be obtained from those by a simple `gauge transformation'.
\end{remark}

\begin{remark}
It must be stressed again that the above calculations were done within one pentachoron~$u$. In particular, the factor~$p$ in \eqref{pf} and ~\eqref{ph} may not be the same for another pentachoron.
\end{remark}

\section[Relations corresponding to moves 0--2, and the factors~$\eta_u$]{Relations corresponding to moves 0--2, and the factors~$\boldsymbol{\eta_u}$}\label{s:02}

\subsection{Grassmann delta functions}\label{ss:Gd}

Quite naturally (see Theorem~\ref{th:m1}), our Grassmann-algebraic relations corresponding to moves 0--2 introduced in Section~\ref{s:P} will involve \emph{Grassmann delta functions}. The following simple definition will suit our needs well enough.

\begin{dfn}\label{dfn:Gd}
Let $\vartheta$ and~$\vartheta'$ be two Grassmann generators. Then we introduce the Grassmann delta function as follows:
\begin{equation}\label{d}
\delta(\vartheta,\vartheta')\stackrel{\mathrm{def}}{=}\vartheta-\vartheta'.
\end{equation}
\end{dfn}

This definition is justified by the following equality:
\begin{equation}\label{de}
\int f(\vartheta) \, \delta(\vartheta,\vartheta') \,\mathrm d\vartheta = f(\vartheta').
\end{equation}
Here $f(\vartheta)$ is any Grassmann algebra element (that may contain Grassmann generator~$\vartheta$ and/or other generators), while $f(\vartheta')=f(\vartheta)[\vartheta/\vartheta']$ is the result of replacing~$\vartheta$ with~$\vartheta'$. Equality~\eqref{de} is easily checked directly if we write $f(\vartheta)$ as $f(\vartheta)=f_0+f_1\vartheta$, where neither of $f_0$ and~$f_1$ contain~$\vartheta$.

\subsection{Relation corresponding to the first move 0--2}\label{ss:02_24}

The first move 0--2 has been described in Subsection~\ref{ss:24}. When we insert two new pentachora~$12456$ in the triangulation this way, there appears one new edge, namely~$b^{\mathrm{new}}=12$. Accordingly, the set~$\mathsf B$ of edges introduced in Subsection~\ref{ss:igb} can be changed the following way:
\[
\mathsf B^{\mathrm{new}}=\mathsf B^{\mathrm{old}}\cup\{ b^{\mathrm{new}} \},
\]
the operator~$\boldsymbol{\partial}$ introduced also in Subsection~\ref{ss:igb} changes as follows:
\[
\boldsymbol{\partial}^{\mathrm{new}}=\boldsymbol{\partial}^{\mathrm{old}}D_{b^{\mathrm{new}}},
\]
and the new expression for~$\boldsymbol{\partial}^{-1} 1$ can be chosen as follows:
\begin{equation}\label{1/d}
(\boldsymbol{\partial}^{-1} 1)^{\mathrm{new}}=(\boldsymbol{\partial}^{-1} 1)^{\mathrm{old}} \mathbf w,
\end{equation}
where $\mathbf w$---it can be called \emph{Grassmann weight of the edge~$b^{\mathrm{new}}$}---is any element of the Grassmann algebra obeying
\[
D_{b^{\mathrm{new}}}\mathbf w = 1,
\]
and, in addition, depending only on Grassmann variables \emph{not present in~$(\boldsymbol{\partial}^{-1} 1)^{\mathrm{old}}$}, that is, only on variables $\vartheta_{1245}$, $\vartheta_{1246}$ and/or~$\vartheta_{1256}$ living on the newly created tetrahedra.

\begin{example}\label{xmp:w1}
For instance, we can choose~$\mathbf w$ to be proportional to~$\vartheta_{1256}$, namely,
\begin{equation}\label{w1}
\mathbf w = \frac{\vartheta_{1256}}{\beta_{12,1256}},
\end{equation}
where $\beta_{12,1256}$ is the coefficient of~$\partial_{1256}$ in the edge operator~$d_{12}$, see the definition~\eqref{I.14c}.
\end{example}

We now introduce the ``pillow Grassmann weight'' for our first move 0--2. Its unnormalized version\footnote{Later we will ``normalize'' it, see the l.h.s.\ of~\eqref{pw24}.} reads as follows:
\begin{equation}\label{P}
\mathcal P = \iiint \mathcal W_{12456} \widetilde{\mathcal W}_{12456}\, \mathbf w\,\mathrm d\vartheta_{1245}\,\mathrm d\vartheta_{1246}\,\mathrm d\vartheta_{1256} .
\end{equation}
Both $\mathcal W_{12456}$ and~$\widetilde{\mathcal W}_{12456}$ are Grassmann--Gaussian exponentials~(I.16), and the differences between them are as follows:
\begin{enumerate}\itemsep 0pt
 \item as they belong to the pentachora with opposite orientations, their respective edge operators have the same differential parts, but their `$\vartheta$-parts' differ in signs,
 \item\label{i:twoF} consequently, the two respective matrices~$F$ also differ in signs\footnote{A bit more detailed explanation: $\gamma_{ij,t}$ in~(I.44) changes its sign together with the orientation of tetrahedron~$t$, hence $\gamma_{tt'}$ in~(I.49) changes its sign together with the orientation of tetrahedron~$t'$, and this latter orientation is induced from the pentachoron.},
 \item additionally, while tetrahedra $1245$, $1246$ and~$1256$ lie inside the pillow and are common for our two pentachora, each of these has its own copies of boundary tetrahedra $1456$ and~$2456$, and we have to introduce special notation for the corresponding Grassmann generators. Namely, we denote as $\vartheta_{1456}$ and~$\vartheta_{2456}$ variables entering in the pentachoron weight~$\mathcal W_{12456}$, while $\vartheta_{1456}'$ and~$\vartheta_{2456}'$ enter in~$\widetilde{\mathcal W}_{12456}$; there are also $\vartheta_{1245}$, $\vartheta_{1246}$ and~$\vartheta_{1256}$ that enter in both weights.
\end{enumerate}

A small exercise based mainly on item~\ref{i:twoF} shows that the weight~\eqref{P} must be proportional to the product of two Grassmann delta functions~\eqref{d}:
\begin{equation}\label{Phi}
\mathcal P = \Phi \,\delta(\vartheta_{1456},\vartheta'_{1456})\, \delta(\vartheta_{2456},\vartheta'_{2456}).
\end{equation}
Coefficient~$\Phi$ can be expressed in terms of values~$\omega_s$ for 2-faces~$s$ of our pentachoron~$12456$ in many equivalent ways; one of these is described in the following Example~\ref{xmp:Phi}.

\begin{example}\label{xmp:Phi}
We calculate the coefficient~$\Phi$ by considering the terms proportional to $\vartheta_{1456}\vartheta'_{2456}$ in both sides of~\eqref{Phi}. Also, we take the edge weight in the form~\eqref{w1}. This means that the only term in the product $\mathcal W_{12456} \widetilde{\mathcal W}_{12456}$ in~\eqref{P} that is important for us now is that proportional to $\vartheta_{1456}\vartheta_{1245}\vartheta'_{2456}\vartheta_{1246}$, and the coefficient of proportionality is
\[
F_{1456,1245}\widetilde F_{2456,1246}-F_{1456,1246}\widetilde F_{2456,1245}.
\]
Here we denoted, just for clearness, matrix elements coming from the weight~$\mathcal W_{12456}$ as~$F_{tt'}$, while those coming from the weight~$\widetilde{\mathcal W}_{12456}$ as~$\widetilde F_{tt'}$; as we have already explained, $\widetilde F_{tt'}=-F_{tt'}$. One can also see from formula~(I.16) that the monomial in~$\mathcal W_u$ proportional to $\vartheta_t \vartheta_{t'}$, for two 3-faces $t,t'\subset u$, is $-F_{tt'}\vartheta_t \vartheta_{t'}$. It follows that one possible expression for~$\Phi$ is
\begin{equation}\label{Phii}
\Phi = \frac{F_{1456,1245}F_{2456,1246}-F_{1456,1246}F_{2456,1245}}{\beta_{12,1256}}.
\end{equation}
\end{example}

As we noted in Subsection~\ref{ss:24}, four new 2-faces, namely $124$, $125$, $126$ and~$456$, appear when our first move 0--2 is done. It turns out that the following quantity:
\[
\mathrm H = \frac{q_{124}q_{125}q_{126}q_{456}}{\Phi},
\]
where the numerator is the product of values~$q_s$ for the mentioned 2-faces, has remarkable symmetry properties. Roughly speaking, $\mathrm H$ is invariant under all permutations of the pentachoron vertices. The exact statement is presented below as Theorem~\ref{th:eta}, but first we introduce one more notational convention.

Recall (Convention~I.1) that all our vertices are always numbered and consequently there is a natural order on them. This order plays an auxiliary---but useful---role in our constructions, as can be seen in Subsection~I.5.2, see especially formula~(I.40). Below in~\eqref{Hu}, $\{ijk\}$ means that $i$, $j$ and~$k$ are taken in this (increasing) order, for instance, $q_{\{635\}}$ means~$q_{356}$.

\begin{theorem}\label{th:eta}
Choose arbitrarily a tetrahedron~$t$ and its edge~$b$, both belonging to a given pentachoron~$u$, that is, $b\subset t\subset u$. Edge $b=ij$ $(\,=-ji)$ is understood as oriented; pentachoron~$u$ is also oriented. The orientation of~$u$ can always be written as given by the order of vertices in $u=ijklm$, where $i$ and~$j$ belong to~$b$, while $m$ is, by definition, the vertex \emph{not} belonging to the tetrahedron~$t$, and the remaining two vertices are denoted $k$ and\/~$l$ in such way as to give the needed orientation.

Then the quantity
\begin{equation}\label{Hu}
\mathrm H_u = \frac{\beta_{bt}\, q_{\{ijk\}} q_{\{ijl\}} q_{\{ijm\}} q_{\{klm\}}} {F_{iklm,ijkm}F_{jklm,ijlm}-F_{iklm,ijlm}F_{jklm,ijkm}}
\end{equation}
remains the same for all pairs $b\subset t$ and thus belongs only to~$u$ itself.
\end{theorem}

\begin{proof}
This Theorem states some algebraic equalities between quantities expressible in terms of the cocycle~$\omega$. So, in principle, they can be proved by a direct calculation. In reality, however, this turned out too hard even for a computer using computer algebra.

Happily, there is a roundabout way. As only one matrix~$F$ is present in~\eqref{Hu}, we can take \emph{its entries as independent variables}, and express everything in their terms. All necessary formulas are written out in Section~\ref{s:oF} and, with them, the proof becomes an easy exercise.
\end{proof}

Recall that we have not yet given the expression for quantities~$\eta_u$. \emph{Define} now $\eta_u$ as follows:
\begin{equation}\label{eta}
\eta_u = \sqrt{\mathrm H_u}.
\end{equation}
This definition contains, due to~\eqref{Hu}, \emph{fourth} roots of values~$\omega_s$. As the signs of~$\omega_s$ already depend on a vertex ordering: $\omega_{ijk}=-\omega_{jik}$ and so on, taking further roots involves such quantities as~$\sqrt{-1}$, and these must be taken under control. This is what the following theorem is about.

\begin{theorem}\label{th:oe}
Quantity~$I$~\eqref{inv}, taken up to a sign\footnote{Recall Remark~\ref{r:+-}.}, does not depend on a vertex ordering, if the values of square roots~$K_t$~(I.36) are fixed\/\footnote{It is understood, of course, that the conditions (I.37) and~(I.38) are also fulfilled. As for the actual \emph{independence} of the signs of~$K_t$, it will be proved in Lemma~\ref{l:rr}.}\!.
\end{theorem}

\begin{proof}
It is enough to consider the interchange of numbers between two vertices with neighboring numbers $i$ and $j=i+1$. Moreover, anything nontrivial may happen only if there is an edge~$ij$ in the triangulation. In this case, all~$q_s$ with $s\supset ij$ are multiplied by the \emph{same} root of~$-1$: this follows from the fact that, for any tetrahedron~$t$, the product $\prod_{s\subset t} q_s$ must give the quantity~$K_t$ corresponding to the orientation of~$t$ determined by the order of vertices, see~\eqref{I.40} (and this orientation changes when the numbers are permuted).

Denote the number of triangles~$s$ around edge~$ij$ as~$n_2$; then the number of pentachora around~$ij$ is $n_4=2n_2-4$ (this is an easy exercise using the fact that the \emph{link} of an edge is a triangulated two-sphere). As a result of our permutation, $\prod_{\mathrm{all}\;s}q_s$ is multiplied by $(\sqrt{-1})^{n_2}$. On another hand, $\eta_u$ for each pentachoron $u\supset ij$ gets multiplied by $\pm(\sqrt{-1})^{3/2}$, according to~\eqref{Hu} and the fact that there are exactly three triangles in~$u$ containing~$ij$. The overall factor for $\prod_{\mathrm{all}\;u}\eta_u$ is thus $\pm(\sqrt{-1})^{3n_2-6}$, and this means that, altogether, \eqref{inv} acquires a factor of $\pm(\sqrt{-1})^{2n_2-6}$, which is~$\pm 1$.
\end{proof}

Formulas \eqref{Phi} and~\eqref{eta} imply the following beautiful identity for the pillow Grassmann weight~$\mathcal P$:
\begin{equation}\label{pw24}
\frac{\eta_{12456}^2}{q_{124}q_{125}q_{126}q_{456}} \mathcal P = \delta(\vartheta_{1456},\vartheta'_{1456})\,\delta(\vartheta_{2456},\vartheta'_{2456}),
\end{equation}
with a transparent meaning of each factor in it. Namely, the numerator~$\eta_{12456}^2$ of the fraction before~$\mathcal P$ corresponds to the fact that two new pentachora~$12456$ have appeared, and the denominator---to the fact that new 2-faces $124$, $125$ and~$126$ have appeared, while 2-face~$456$ has doubled.

\begin{theorem}\label{th:m1}
Expression~\eqref{inv} remains the same under our ``first move 0--2''.
\end{theorem}

\begin{proof}
Due to \eqref{1/d} and~\eqref{pw24}, and using the fundamental property~\eqref{de} of delta function, we can express the integral in~\eqref{23GB}, taken after the move 0--2, as the same integral taken before this move and multiplied by $\frac{q_{124}q_{125}q_{126}q_{456}}{\eta_{12456}^2}$. Then it follows from \eqref{23GB} and~\eqref{rt} that the value~$I$~\eqref{inv} remains the same.
\end{proof}

\subsection{Relation corresponding to the second move 0--2}\label{ss:02_15}

The second move 0--2 has been described in Subsection~\ref{ss:15}. This time, we insert two pentachora~$13456$ in the triangulation, and four new edges are added, namely $13$, $14$, $15$ and~$16$. Any three of them can be added to the set~$\mathsf B$:
\[
\mathsf B^{\mathrm{new}}=\mathsf B^{\mathrm{old}}\cup \{b_1^{\mathrm{new}},b_2^{\mathrm{new}},b_3^{\mathrm{new}}\}.
\]
Accordingly, we choose
\[
(\boldsymbol{\partial}^{-1} 1)^{\mathrm{new}}=(\boldsymbol{\partial}^{-1} 1)^{\mathrm{old}} \mathbf w,
\]
where this time $\mathbf w$ satisfies
\[
D_{b_1^{\mathrm{new}}}D_{b_2^{\mathrm{new}}}D_{b_3^{\mathrm{new}}}\mathbf w=1
\]
and depends only on variables~$x_t$ on newly created tetrahedra~$t$.

For instance, take $b_1^{\mathrm{new}}=13$, \ $b_2^{\mathrm{new}}=14$, \ $b_3^{\mathrm{new}}=15$. Then $\mathbf w$ can be taken proportional to $\vartheta_{1345} \vartheta_{1346} \vartheta_{1356}$, and the proportionality factor is $(\det m)^{-1}$, where $m$ is the matrix of coefficients of $\partial_{1345}$, $\partial_{1346}$ and~$\partial_{1356}$ in edge operators $d_{13}$, $d_{14}$ and~$d_{15}$:
\[
m = \begin{pmatrix} \beta_{13,1345} & \beta_{13,1346} & \beta_{13,1356} \\
                    \beta_{14,1345} & \beta_{14,1346} & 0 \\
                    \beta_{15,1345} & 0 & \beta_{15,1356} \end{pmatrix}.
\]
Recall that formulas for coefficients~$\beta_{bt}$ can be found in Subsection~I.5.3.

We would like to have now the following equality for the ``pillow weight'':
\begin{multline}\label{pw15}
\frac{\eta_{13456}^2}{q_{134}q_{135}q_{136}q_{145}q_{146}q_{156}} \iiiint \mathcal W_{13456} \widetilde{\mathcal W}_{13456}\, \mathbf w\,\mathrm d\vartheta_{1345}\,\mathrm d\vartheta_{1346}\,\mathrm d\vartheta_{1356}\,\mathrm d\vartheta_{1456} \\
=\delta(\vartheta_{3456}, \vartheta_{3456}'),
\end{multline}
where $\vartheta_{3456}$ and~$\vartheta_{3456}'$ correspond to the two copies of tetrahedron~$3456$, having Grassmann weights $\mathcal W_{13456}$ and $\widetilde{\mathcal W}_{13456}$, respectively.

Indeed, the quadruple integral in~\eqref{pw15} is easily evaluated to
\begin{equation}\label{Psi*}
-\frac{F_{3456,1456}}{\det m}\, \delta(\vartheta_{3456}, \vartheta_{3456}'),
\end{equation}
where the delta function appears, much like in Subsection~\ref{ss:02_24}, from integrating the product of two exponentials of quadratic forms differing in signs. The miracle is the following.

\begin{theorem}\label{th:e2}
The inverse to the coefficient of delta function in~\eqref{Psi*} can again be written in terms of $\eta_u$ and~$q_s$ according to the same principle as in~\eqref{pw24}, that is, as a product of $\eta_u$ for each new pentachoron~$u$, and $q_s^{-1}$ for each new 2-face. Hence, formula~\eqref{pw15} holds indeed.
\end{theorem}

\begin{proof}
Take, like we did in the proof of Theorem~\ref{th:eta}, the entries of~$F$ as independent variables. This makes the proof feasible by a direct computer calculation.
\end{proof}

\begin{theorem}\label{th:m2}
Expression~\eqref{inv} remains the same under our ``second move 0--2''.
\end{theorem}

\begin{proof}
This invariance of~\eqref{inv} is proved in full analogy with Theorem~\ref{th:m1}.
\end{proof}

\subsection[Quantities~$\eta_u$ and move 3--3]{Quantities~$\boldsymbol{\eta_u}$ and move 3--3}\label{ss:ER}

One more miracle is that the quantities~$\eta_u$ introduced in Subsection~\ref{ss:02_24}, work also for the Pachner move 3--3.

\begin{conjecture}
Relation~\eqref{33} holds indeed, and with the same~$\eta_u$~\eqref{eta}.
\end{conjecture}

As we have explained in the Introduction, this Conjecture is actually a firmly established mathematical fact, although not yet formally proven.

\section[Dependence of~$I$ only on~$M$ and the cohomology class of~$\omega$]{Dependence of~$\boldsymbol{I}$ only on~$\boldsymbol{M}$ and the cohomology class of~$\boldsymbol{\omega}$}\label{s:f}

We will be able to call our quantity $I$~\eqref{inv} the invariant of a pair ``piecewise linear manifold~$M$, cohomology class~$h\ni \omega$\,''---assuming of course our Conjecture in Subsection~\ref{ss:ER}, and keeping in mind Remark~\ref{r:+-}---if we check its invariance under everything that may be changed in out calculations. Namely, $I$ must be independent of
\begin{enumerate}\itemsep 0pt
 \item\label{i:a} a specific triangulation,
 \item\label{i:b} vertex ordering,
 \item\label{i:c} permitted signs of~$K_t$~(I.36),
 \item\label{i:d} and choice of~$\omega$ within its cohomology class.
\end{enumerate}

Item~\ref{i:a} is solved using our formulas corresponding to Pachner moves---this has been the main subject matter in the previous Sections.

Item~\ref{i:b} has been solved in Theorem~\ref{th:oe}.

The two remaining items are solved below in Lemma~\ref{l:rr} and Theorem~\ref{th:inv}.

\begin{lemma}\label{l:rr}
Assuming our Conjecture, the quantity~$I$ is the same for any permitted (see (I.37) and~(I.38)) choice of signs of~$K_t$.
\end{lemma}

\begin{proof}
Permitted way of changing the signs of some~$K_t$ consists in changing the signs of some~$q_s$, due to Assumption~I.1, see Theorem~I.5. Consider just one~$q_s$, and imagine a sequence of Pachner moves where the last of them removes 2-face~$s$ from the triangulation. Then do this sequence backwards, beginning with inserting~$s$ and the corresponding~$q_s$ with the \emph{changed} sign. This can always be done, because our formulas \eqref{33}, \eqref{pw24} and~\eqref{pw15} hold for any choice of the signs of~$q_s$.
\end{proof}

\begin{theorem}\label{th:inv}
Assuming our Conjecture, our quantity~$I$~\eqref{inv} remains the same for all cocycles~$\omega$ within a given cohomology class~$h$ and is thus indeed an invariant of the pair $(M,h)$.
\end{theorem}

\begin{proof}
Recall that we agreed (in Section~I.3) to denote a 1-cochain taking value~$1$ on an edge~$b$ and vanishing on all other edges, simply by the same letter~$b$. Any 2-\emph{coboundary} is a linear combination of edge coboundaries~$\delta b$. We are going to show how to change
\begin{equation}\label{co}
\omega \mapsto \omega + c\,\delta b
\end{equation}
for any edge~$b$ and number~$c$, without changing~$I$.

First, we recall that our pillow Grassmann weights, taken with the corresponding multipliers, are simply delta functions, see \eqref{pw24} and~\eqref{pw15}. Hence, they do \emph{not} depend on the cocycle~$\omega$. If there is such pillow within a triangulation, and we change~$\omega$ by a multiple of~$\delta b$ for any edge~$b$ lying \emph{inside} the pillow, this will not affect~$I$ either.

Second, there always exists a sequence of Pachner moves whose last move 4--2 or 5--1 takes $b$ away from the triangulation. Now we do such sequence, and then its inverse, returning to the initial triangulation. But when we do the first move 2--4 or 1--5 in the inverse sequence, we represent it as a composition of moves where a move 0--2 is the first (according to Subsection~\ref{ss:24} or~\ref{ss:15}), and, while doing this 0--2, change $\omega$, with respect to its initial values, according to~\eqref{co}.
\end{proof}

\end{document}